\documentclass[12pt,reqno,fleqn]{amsart}
\usepackage{latexsym}
\usepackage{amssymb}
\usepackage{amsxtra}
\usepackage{amsmath}
\usepackage{amsfonts}

\usepackage{graphicx}
\usepackage{subfigure}
\usepackage[section]{placeins}
\usepackage{psfrag}
\usepackage{caption2}
\usepackage{flafter}
\usepackage{epsfig}
\usepackage{epsfig}
\usepackage{bm}%

\def\fnum@figure{\figurename\thefigure}
\newcommand{\cleqn}{\setcounter{equation}{0}}
\newcommand{\clth}{\setcounter{theorem}{0}}
\newcommand {\sectionnew}[1]{\section{#1}\cleqn\clth}

\makeatletter      
\@addtoreset{equation}{section}
\makeatother       

\newtheorem{theorem}{Theorem}[section]
\newtheorem{lemma}[theorem]{Lemma}

\newtheorem{definition}[theorem]{Definition}

\renewcommand{\P}{\mathcal{P}}

\renewcommand{\L}{\mathcal{L}}

\newcommand{\Z}{\mathbb{Z}}

\newcommand{\T}{\mathbb{T}}

\def\({\left(}
\def\){\right)}
\newcommand{\La}{\Lambda}
\def\d{\partial}
\newcommand{\ep}{\epsilon}

\renewcommand{\L}{\mathcal L}
\def\({\left(}
\def\){\right)}
\def\[{\begin{eqnarray}}
\def\]{\end{eqnarray}}
\begin{document}

\title{Multi-fold Darboux transformations of the extended bigraded Toda Hierarchy}
\author{
Chuanzhong Li, Tao Song}
\allowdisplaybreaks
\dedicatory {\small Department of Mathematics,  Ningbo University, Ningbo, 315211, China\\
 Email:lichuanzhong@nbu.edu.cn}
\thanks{}
\texttt{}

\date{}

\begin{abstract}
 With the extended  logarithmic flow equations and some extended Vertex operators in generalized Hirota bilinear equations, extended bigraded  Toda hierarchy(EBTH) was proved to govern the Gromov-Witten theory of orbiford $c_{NM}$ in literature. The generating function of these Gromov-Witten invariants is one special solution of the EBTH. In this paper, the multi-fold Darboux transformations and their determinant representations of the EBTH are given with two different gauge transformation operators. The two Darboux transformations in different directions are used to generate new solutions from known solutions which include soliton solutions of $(N,N)$-EBTH, i.e. the EBTH when $N=M$. From the generation of new solutions, one can find the big difference between the EBTH and the extended Toda hierarchy(ETH). Meanwhile we plotted the soliton graphs of the $(N,N)$-EBTH from which some approximation analysis will be given. From the analysis on velocities of soliton solutions, the difference between the extended flows and other flows are shown. The two different Darboux transformations constructed by us might be useful in Gromov-Witten theory of orbiford $c_{NM}$.
\end{abstract}


\maketitle

Mathematics Subject Classifications(2000).  37K10, 37K20.\\
Keywords:  extended bigraded Toda hierarchy,  Darboux transformations, determinant representation, soliton solution.\\
\allowdisplaybreaks
\tableofcontents
 \setcounter{section}{0}

\sectionnew{Introduction}
The Toda lattice equation is a nonlinear evolutionary
differential-difference equation introduced by Toda \cite{Toda}
describing an infinite system of masses on a line that interact
through an exponential force. This equation was further generalized to Toda lattice hierarchy \cite{UT}. It is completely integrable, i.e. admits infinite
conserved quantities and Lax pair. It has important applications in many
different fields such as classical and quantum field theory, in
particular in the theory of Gromov-Witten invariants (\cite{Z}).
 Considering its application to 2D topological field theory (\cite{D witten},
\cite{witten}) and string theory (\cite{dubrovin}), one replaced the
discrete variables with continuous one. After continuous ``
interpolation'' (\cite{CDZ}) to the whole Toda lattice hierarchy, it
was found that the flow of the spatial translation was missing. In order to
get a complete family of flows,  the interpolated Toda lattice
hierarchy was extended into the so-called extended Toda hierarchy
(\cite{CDZ}). It was first conjectured and then shown (\cite{DZ},
\cite{Ge}, \cite{Z}) that the extended Toda hierarchy is the
hierarchy describing the Gromov-Witten invariants of $CP^1$ by
matrix models (\cite{matrix model}) which describe in the large $N$
limit of the $CP^1$ topological sigma model. The Darboux transformation and soli
ton solution of the extended Toda hierarchy was shortly discussed in \cite{carletthesis}.  The extended bigraded
Toda hierarchy (EBTH) was introduced by Gudio Carlet (\cite{C}) who
hoped that EBTH might also be relevant for some applications in 2D
topological field theory and in the theory of Gromov-Witten
invariants. In the paper (\cite{C}), he generalized the Toda lattice
hierarchy and extended Toda lattice hierarchy by considering $N+M$
dependent variables and used them to provide a Lax pair definition
of the extended bigraded Toda hierarchy. In the paper of Todor E.
Milanov and Hsian-Hua Tseng (\cite{TH}), they described
conjecturally a kind of  Hirota bilinear equations (HBEs) which was
similar to the Lax operators of  the EBTH and proved that it
governed the Gromov-Witten theory of orbiford $c_{NM}$.
The Hirota bilinear equation of EBTH were equivalently constructed in our early paper \cite{ourJMP} and a very recent paper \cite{leurhirota}, because of the equivalence of $t_{1,N}$ flow and $t_{0,N}$ flow of the EBTH in \cite{ourJMP}. Meanwhile it was proved to govern Gromov-Witten invariant of the total
  descendent potential of $\mathbb{P}^1$ orbifolds \cite{leurhirota}. This
hierarchy also lead to a series of results from analytical and algebraic considerations
\cite{ourBlock,solutionBTH,RMP,annal}. However the explicit solutions of the EBTH are still unknown because the constraint on spectral parameter is complicated and most of equations are nonlocal. Also how the extended flows change the velocities of the solutions is still not clear. Therefore the purpose of this paper is to solve the EBTH and identify the affection of the extended flows of the EBTH on its solutions.

Among many analytical methods, it is well known that the Darboux
 transformation is one of the efficient methods to generate the
 soliton solutions for integrable systems \cite{Liys,matveev,Mateev,binamatveev,Oevel}.
 In \cite{adler92}, the two Darboux transforms on band matrices called $LU$  and $UL$ Darboux transformations are constructed, particularly for the $2m+1$-band matrix. In fact the $2m+1$-band matrix corresponds to the $(m,m)$-EBTH. The   $LU$  and $UL$ Darboux transformations inspire us to consider two separated Darboux transformations of the EBTH.  The two separated Darboux transformations of the EBTH show that there exists a big difference between the EBTH and the ETH.
 The determinant representation of $n$-fold Darboux transformation gives a convenient tool to explicitly express new solutions\cite{gaugetras,Hedeterminant,PRE}. This remind us to consider the  Darboux transformation and its determinant representation of the EBTH. This will be used to generate new solutions from known solutions which include soliton solutions and solutions related to the Gromov-Witten theory of orbiford $c_{NM}$.

The paper is organized as follows. In Section 2 we recall the roots
and the logarithms of the Lax operator $\L$ and  the definition
of EBTH.    In Section 3, the $n$-th Darboux transformation and its determinant representation of the EBTH with the help of the first wave function is given which is used to generate new solutions from seed solutions which include soliton solutions in Section 4. Using the second wave function, the second Darboux transformation will be constructed which benefits in showing the character of the $t_{\beta,n}, -M+1\leq \beta\leq 0$ flow in Section 5. Meanwhile the second kind of soliton solutions can be given using the second Darboux transformation. Section 7 will be devoted to conclusions and discussions.

\sectionnew{ The extended  bigraded Toda Hierarchy }
We describe the lax form of the EBTH following \cite{C,ourJMP}. Introduce
firstly the lax operator
\begin{equation}\L=\Lambda^{N}+u_{N-1}\Lambda^{N-1}+\dots + u_{-M}
\Lambda^{-M}
\end{equation}
(for $N,M \geq1$ are two fixed positive integers and $u_{-M}$ is a
non-vanishing function).
  The variables $u_j$ are functions of the real continuous variable $x$ and the
  shift operator $\Lambda$ acts on a function $a(x)$ by $\Lambda a(x) = a(x + \epsilon
  )$,
i.e. $\Lambda$ is equivalent to  $e^{\epsilon\partial_{x}}$ where
the spacing unit $``\epsilon"$ is called string coupling constant.
The Lax operator $\L$ can be written in two different ways by
dressing the shift operator
  \begin{eqnarray}
  \label{dressing}\L=\P_L\Lambda^N\P_L^{-1} = \P_R \Lambda^{-M}\P_R^{-1}.\end{eqnarray}
The two dressing operators have the following form \begin{eqnarray}
&& \P_L=1+w_1\Lambda^{-1}+w_2\Lambda^{-2}+\ldots,
\label{dressP}\\
&& \P_R=\tilde{w_0}+\tilde{w_1}\Lambda+\tilde{w_2}\Lambda^2+ \ldots,
\label{dressQ} \end{eqnarray}
where $\tilde{w_0}$ is not zero.
  From identity \eqref{dressing}, we
can easily get the relation of $u_i$ and $w_j$ as following
\begin{eqnarray}
  u_{N-1}&=&w_1(x)-w_1(x+N\epsilon) \\
  u_{N-2}&=&w_2(x)-w_2(x+N\epsilon)-(w_1(x)-w_1(x+N\epsilon))w_1(x+(N-1)\epsilon) \\\notag
   u_{N-3}&=&w_3(x)-w_3(x+N\epsilon)- [w_2(x)-w_2(x+N\epsilon)
  -(w_1(x)-w_1(x+N\epsilon))w_1(x+(N-1)\epsilon)]\\
  &&w_1(x+(N-2)\epsilon)-(w_1(x)-w_1(x+N\epsilon))w_2(x+(N-3)\epsilon) \\
  \notag \cdots &\cdots &\cdots.
\end{eqnarray}\\
From identity \eqref{dressing}, we can also easily get the relation of $u_i$ and
$\tilde w_j$ formally as following
\begin{eqnarray}
  u_{-M}&=&\frac{\tilde w_0(x)}{\tilde w_0(x-M\epsilon)} \\
  u_{-M+1}&=&\frac{\tilde w_1(x)-\frac{\tilde w_0(x)}{\tilde w_0(x-M\epsilon)}\tilde w_1(x-M\epsilon)}{\tilde w_0(x-(M-1)\epsilon)} \\\notag
   u_{-M+2}&=&\frac{\tilde w_2(x)-\frac{\tilde w_0(x)}{\tilde w_0(x-M\epsilon)}\tilde w_2(x-M\epsilon)-\frac{\tilde w_1(x)-\frac{\tilde w_0(x)}
   {\tilde w_0(x-M\epsilon)}\tilde w_1(x-M\epsilon)}{\tilde w_0(x-(M-1)\epsilon)}\tilde w_1(x-(M-1)\epsilon)}{\tilde w_0(x-(M-2)\epsilon)}\\
  &&
\\
  \notag \cdots &\cdots &\cdots\\\notag
 u_{N-1}&=&\frac{\tilde w_{M+N-1}-u_{-M}\tilde w_{M+N-1}(x-M\epsilon)
   -\dots-u_{N-2}\tilde w_1(x+(N-2)\epsilon)}{\tilde w_0(x+(N-1)\epsilon)}\\\notag
 u_N=  1&=&\frac{\tilde w_{M+N}-u_{-M}\tilde w_{M+N}(x-M\epsilon)
   -\dots-u_{N-1}\tilde w_1(x+(N-1)\epsilon)}{\tilde
   w_0(x+N\epsilon)}.
\end{eqnarray} One can define the fractional powers $\L^{\frac1N}$ and
$\L^{\frac1M}$ in the form of
\begin{equation}
  \notag
  \L^{\frac1N} = \Lambda+ \sum_{k\leq 0} a_k \Lambda^k , \qquad \L^{\frac1M} = \sum_{k \geq -1} b_k \Lambda^k
\end{equation}
defined by the relations
\begin{equation}
  \notag
  (\L^{\frac1N} )^N = \L, \qquad (\L^{\frac1M} )^M = \L.
\end{equation}
It should be stressed that $\L^{\frac1N}$ and $\L^{\frac1M}$ are two different
operators even if $N=M$ because of two fractional expansions in two different directions.

Of course an equivalent definition can be given in terms of the
dressing operators
\begin{equation}
  \notag
  \L^{\frac1N} = \P _{L}\Lambda\P_{L}^{-1}, \qquad \L^{\frac1M} = \P_{R}\Lambda^{-1} \P_{R}^{ -1}.
\end{equation}

We also define two logarithms of the operator $\L$ by the
following formulas as \cite{C,ourJMP}

\begin{subequations}  \notag
 \begin{align}
  &\log_+ \L = \P_{L} N \epsilon \partial \P_{L}^{-1} = N \epsilon \partial - N \epsilon \P_{Lx} \P_{L}^{-1}, \\
  &\log_- \L = - \P_{R}M \epsilon\partial\P_{R}^{-1} = - M \epsilon \partial + M \epsilon \P_{Rx} \P_{R}^{-1},
\end{align}
\end{subequations}

where $\partial = \frac{d}{dx}$. These are differential-difference
operators in forms as

\begin{align}
   \notag
  &\log_+ \L = N \epsilon \partial + 2 N \sum_{k > 0} W_{-k}(x) \Lambda^{-k} ,\\
  &\log_- \L = -M \epsilon \partial + 2 M \sum_{k \geq 0} W_k(x) \Lambda^k .
  \notag
\end{align}
One can combine them into a single logarithmic operator\cite{C}

\begin{equation}
   \notag
  \log\L = \frac1{2N} \log_+\L + \frac1{2M} \log_- \L = \sum_{k \in \Z} W_k \Lambda^k.
\end{equation}

That is a pure difference operator since the derivatives cancel.
Given any difference operator $A= \sum_k A_k \Lambda^k$, the
positive and negative projections are given by $A_+ = \sum_{k\geq0}
A_k \Lambda^k$ and $A_- = \sum_{k<0} A_k \Lambda^k$.

Similarly to \cite{C}, we give the following definition.
\begin{definition} \label{deflax}
The Lax formulation of the extended bigraded Toda hierarchy is given
by
\begin{equation}
  \label{edef}
\frac{\partial \L}{\partial t_{\alpha, n}} = [ A_{\alpha,n} ,\L ]
\end{equation}
for $\alpha = N-1,N-2, \dots, -M$ and $n \geq 0$. The operators
$A_{\alpha ,n}$ are defined by
\begin{subequations}
\label{Adef}
\begin{align}
  &A_{\alpha,n} = \frac{\Gamma (2- \frac{\alpha}{N} )}{  \epsilon \Gamma(n+2 -\frac{\alpha}{N} ) } ( \L^{n+1-\frac{\alpha}N })_+, \quad \text{for} \quad \alpha = N-1, \dots, 0,\\
  &A_{\alpha,n} = \frac{-\Gamma (2+\frac{\alpha}{M} )}{  \epsilon \Gamma(n+2 +\frac{\alpha}{M} ) } ( \L^{n+1+\frac{\alpha}M })_-, \quad \text{for} \quad \alpha = 0, \dots, -M+1, \\
  &A_{-M,n} = \frac{2}{\epsilon n!} [ \L^n (\log \L - \frac12 ( \frac1M + \frac1N) c_n ) ]_+ .
\end{align}
\end{subequations}
These constants $c_n$ are defined by
\begin{equation}
  \label{b24}
  c_n = \sum_{k=1}^n \frac1k ,  c_0=0 .
\end{equation}
\end{definition}

\subsection{The $(2,2)$-EBTH}

 The Lax operator of the (2,2)-EBTH is
\begin{equation}L=\Lambda^2+u_{1}\Lambda+u_0 + u_{-1}
\Lambda^{-1}+ u_{-2}\Lambda^{-2}.
\end{equation}

Then by Lax equations, we get the $t_{1,0}$ flow of the (2,2)-EBTH
\begin{eqnarray}
 \partial_{t_{1,0}} L= [\Lambda +(1+\La)^{-1}u_1(x), L]
\end{eqnarray}
which correspond to
\begin{eqnarray}\label{10flow}
\begin{cases}
 \partial_{t_{1,0}} u_1(x)&= u_0(x+\epsilon)-u_0(x)+u_1(x)(1-\La)(1+\La)^{-1}u_1(x)\\
 \partial_{t_{1,0}} u_0(x)&= u_{-1}(x+\epsilon)-u_{-1}(x)\\
 \partial_{t_{1,0}} u_{-1}(x)&= u_{-2}(x+\epsilon)-u_{-2}(x)+u_{-1}(x)(1-\La^{-1})(1+\La)^{-1}u_1(x)\\
  \partial_{t_{1,0}} u_{-2}(x)&=u_{-2}(x)(1-\La^{-2})(1+\La)^{-1}u_1(x).
 \end{cases}
 \end{eqnarray}

 The $t_{0,0}$ flow will have finite terms as following because it does not use the fraction power of Lax operator $L$,
\begin{eqnarray}
 \partial_{t_{0,0}} L= [\Lambda^2+u_1\Lambda +u_0, L]
\end{eqnarray}
which correspond to
\begin{eqnarray}\label{00flow}
\begin{cases}
\partial_{t_{0,0}} u_1(x)&= u_{-1}(x+2\epsilon)-u_{-1}(x)\\
 \partial_{t_{0,0}} u_0(x)&= u_{-2}(x+2\epsilon)-u_{-2}(x)+u_1(x)u_{-1}(x+\epsilon)-u_{-1}(x)u_1(x-\ep)\\
 \partial_{t_{0,0}} u_{-1}(x)&= u_1(x)u_{-2}(x+\epsilon)-u_{-2}(x)u_1(x-2\ep)+u_{-1}(x)(u_0(x)-u_0(x-\epsilon))\\
  \partial_{t_{0,0}} u_{-2}(x)&=u_{-2}(x)( u_0(x)-u_0(x-2\epsilon)).
   \end{cases}
\end{eqnarray}
For $t_{-1,0}$ flow, equations will also be complicated because of another fraction power of $L$. The equation is
\begin{eqnarray}
 \partial_{t_{-1,0}} L= -[e^{(1+\Lambda^{-1})^{-1}\log u_{-2}}\Lambda^{-1} , L]
\end{eqnarray}
which corresponds to
\begin{eqnarray}\label{-10flow}
\begin{cases}
\partial_{t_{-1,0}} u_1(x)&= e^{(1+\Lambda^{-1})^{-1}\log u_{-2}(x+2\epsilon)}-e^{(1+\Lambda^{-1})^{-1}\log u_{-2}(x)}\\[1.5ex]
 \partial_{t_{-1,0}} u_0(x)&= u_1(x)e^{(1+\Lambda^{-1})^{-1}\log u_{-2}(x+\epsilon)}-e^{(1+\Lambda^{-1})^{-1}\log u_{-2}(x)}u_1(x-\ep)\\[1.5ex]
 \partial_{t_{-1,0}} u_{-1}(x)&= e^{(1+\Lambda^{-1})^{-1}\log u_{-2}(x)}(u_0(x)-u_0(x-\epsilon))\\[1.5ex]
  \partial_{t_{-1,0}} u_{-2}(x)&=u_{-1}(x)e^{(1+\Lambda^{-1})^{-1}\log u_{-2}(x-\epsilon)}-e^{(1+\Lambda^{-1})^{-1}\log u_{-2}(x)}u_{-1}(x-\epsilon).
  \end{cases}
\end{eqnarray}

The  $t_{-2,0}$ flow equation  of extended bigraded Toda hierarchy is given
by

\begin{equation}\label{-20flow}
\frac{\partial \L}{\partial t_{-2,0}} = \frac{\partial \L}{\partial x}.
\end{equation}
One can find the $t_{1,0}$ flow and the $t_{-1,0}$ flow of the $(2,2)$-EBTH are nonlocal equations.

 For the convenience, we
will  define the following operators:

\begin{equation}
  B_{\alpha , n} :=
\begin{cases}
  \frac{\Gamma ( 2- \frac{\alpha}{N} )}{\epsilon\Gamma (n+2 - \frac{\alpha}{N} )} \L^{n+1-\frac{\alpha}{N}} &\alpha=N-1\dots 1\\
  \frac{\Gamma ( 2 + \frac{\alpha}{M} )}{\epsilon\Gamma (n+2 + \frac{\alpha}{M} )}  \L^{n+1+\frac{\alpha}{M}} &\alpha = 0\dots -M+1\\
 \frac{2}{\epsilon n!} [ \L^n( \log \L - \frac{1}{2}( \frac{1}{M} + \frac{1}{N} ) c_n) ] & \alpha = -M.
  \end{cases}
\end{equation}

In fact the EBTH system can also be equivalently rewritten in form of the following linear differential system
\[\label{linearequation}
\begin{cases}
 \L\phi&=\lambda\phi,\\
\frac{\partial \phi}{\partial t_{\alpha, n}}&=(B_{\alpha,n})_+\phi,\ \alpha = N-1,N-2, \dots, -M, n \geq 0.\end{cases} \]
We will call the function $\phi$ in eq.\eqref{linearequation} the first wave function of the EBTH.
 Particularly for $N=M=1$ this hierarchy
coincides with the extended Toda hierarchy introduced in
\cite{CDZ}.

In the next section, it is time to introduce the Darboux transformation of the EBTH basing on linear equation eq.\eqref{linearequation}.

\section{The first Darboux transformation of the EBTH}
In this section, we will consider the Darboux transformation of the EBTH on Lax operator
 \[\L=\Lambda^{N}+u_{N-1}\Lambda^{N-1}+\dots + u_{-M}
\Lambda^{-M},\]
 i.e.
  \[\label{1darbouxL}\L^{[1]}=\Lambda^{N}+u_{N-1}^{[1]}\Lambda^{N-1}+\dots + u_{-M}^{[1]}
\Lambda^{-M}=W\L W^{-1},\]
where $W$ is the Darboux transformation operator.

That means after Darboux transformation, the spectral problem

\[\L\phi=\Lambda^{N}\phi+u_{N-1}\Lambda^{N-1}\phi+\dots + u_{-M}
\Lambda^{-M}\phi=\lambda\phi,\]
will become

\[\L^{[1]}\phi^{[1]}=\Lambda^{N}\phi^{[1]}+u_{N-1}^{[1]}\Lambda^{N-1}\phi^{[1]}+\dots + u_{-M}^{[1]}
\Lambda^{-M}\phi^{[1]}=\lambda\phi^{[1]}.\]

To keep the Lax pair eq.\eqref{edef} of the EBTH invariant, i.e.

\begin{equation}
\frac{\partial \L}{\partial t_{\alpha, n}} = [ (B_{\alpha,n})_+ ,\L ],\ \
\frac{\partial \L^{[1]}}{\partial t_{\alpha, n}} = [ (B_{\alpha,n}^{[1]})_+ ,\L^{[1]} ],\ \ B_{\alpha,n}^{[1]}:=B_{\alpha,n}(\L^{[1]}),
\end{equation} dressing operator $W$ should satisfy the following dressing equation
\[W_{t_{\gamma,n}}=-W(B_{\gamma,n})_++(WB_{\gamma,n}W^{-1})_+W,\ \ -M \leq \gamma\leq N-1, n\geq 0.\]
where $W_{t_{\gamma,n}}$ means the derivative of $W$ by $t_{\gamma,n}.$
To give the Darboux transformation, we need the following lemma.
\begin{lemma}\label{lema}
The operator $B:=\sum_{n=0}^{\infty}b_n\La^n$ is a non-negative difference operator,  $C:=\sum_{n=1}^{\infty}c_n\La^{-n}$ is a negative difference operator and $f,g$ (short for $f(x),g(x)$) are two functions of spatial parameter $x$, following identities hold
\begin{equation}\label{Bneg}
(Bf \frac{\Lambda^{-1}}{1-\Lambda^{-1}} g)_-=B(f) \frac{\Lambda^{-1}}{1-\Lambda^{-1}} g,\ \ \ (f \frac{\Lambda^{-1}}{1-\Lambda^{-1}} gB)_-=f \frac{\Lambda^{-1}}{1-\Lambda^{-1}}B^*(g),
\end{equation}
\begin{equation}\label{Cpos}
(Cf \frac{1}{1-\Lambda} g)_+=C(f)\frac{1}{1-\Lambda} g,\ \ \ (f \frac{1}{1-\Lambda} gC)_+=f \frac{1}{1-\Lambda}C^*(g).
\end{equation}
\end{lemma}
\begin{proof}
Here we only give the proof of the eq.\eqref{Bneg} by direct calculation
\[\notag
(Bf \frac{\Lambda^{-1}}{1-\Lambda^{-1}} g)_-&=&\sum_{m=0}^{\infty}b_m(f(x+m\epsilon)\La^m \frac{\Lambda^{-1}}{1-\Lambda^{-1}} g)_-\\ \notag
&=&\sum_{m=0}^{\infty}b_mf(x+m\epsilon)(\frac{\Lambda^{m-1}}{1-\Lambda^{-1}})_- g\\ \notag
&=&\sum_{m=0}^{\infty}b_mf(x+m\epsilon)\frac{\Lambda^{-1}}{1-\Lambda^{-1}} g\\
&=&B(f) \frac{\Lambda^{-1}}{1-\Lambda^{-1}} g,\]

\[\notag
(f \frac{\Lambda^{-1}}{1-\Lambda^{-1}} gB)_-&=&\sum_{m=0}^{\infty}(f \frac{\Lambda^{-1}}{1-\Lambda^{-1}} gb_m\La^m)_-\\ \notag
&=&\sum_{m=0}^{\infty}(f \frac{\Lambda^{-1}}{1-\Lambda^{-1}}\La^m g(x-m\epsilon)b_m(x-m\epsilon))_-\\ \notag
&=&\sum_{m=0}^{\infty}f( \frac{\Lambda^{m-1}}{1-\Lambda^{-1}})_- g(x-m\epsilon)b_m(x-m\epsilon)\\ \notag
&=&\sum_{m=0}^{\infty}f \frac{\Lambda^{-1}}{1-\Lambda^{-1}} b_m(x-m\epsilon)g(x-m\epsilon)\\
&=&f \frac{\Lambda^{-1}}{1-\Lambda^{-1}}B^*(g).\]
Similar proof for the eq.\eqref{Cpos} can be got easily.

\end{proof}

Now, we will give the following important theorem which will be used to generate new solutions.

\begin{theorem}
If $\phi$ is the first wave function of the EBTH,
the Darboux transformation operator of the EBTH
 \[W(\lambda)=(1-\frac{\phi}{\La^{-1}\phi}\La^{-1})=\phi\circ(1-\La^{-1})\circ\phi^{-1},\]

will generater new solutions $u_{i}^{[1]}$ from seed solutions
$u_{i}, -M\leq i\leq N-1$

\[\label{1uN-11}u_{N-1}^{[1]}&=&u_{N-1}+(\La^N-1)\frac{\phi}{\La^{-1}\phi},\\ \notag
\dots&& \dots\\
u_{i}^{[1]}&=&u_{i}+u_{i+1}^{[1]}(\La^{i+1}\frac{\phi}{\La^{-1}\phi})-\frac{\phi}{\La^{-1}\phi}(\La^{-1}u_{i+1}),\\ \notag
\dots&& \dots,\\ \label{1u-M-11}
u_{-M}^{[1]}&=&\frac{\phi}{\La^{-1}\phi}(\La^{-1}u_{-M})\frac{\La^{-M-1}\phi}{\La^{-M}\phi}.\]

\end{theorem}
\begin{proof}

In the following proof, using eq.\eqref{Bneg} in Lemma \ref{lema}, a direct computation will lead to the following
 \begin{eqnarray*}W_{t_{\gamma,n}}W^{-1}&=&(\phi\circ(1-\La^{-1})\circ\phi^{-1})_{t_{\gamma,n}}\phi\circ(1-\La^{-1})^{-1}\circ\phi^{-1}\\
 &=&(((B_{\gamma,n})_+\phi)\circ(1-\La^{-1})\circ\phi^{-1})\phi\circ(1-\La^{-1})^{-1}\circ\phi^{-1}\\
 &&-\phi\circ(1-\La^{-1})\circ((B_{\gamma,n})_+\phi)\phi^{-1}\circ(1-\La^{-1})^{-1}\circ\phi^{-1}\\
 &=&((B_{\gamma,n})_+\phi)\phi^{-1}-\phi\circ(1-\La^{-1})\circ((B_{\gamma,n})_+\phi)\phi^{-1}\circ(1-\La^{-1})^{-1}\circ\phi^{-1}\\
   &=&-(\phi\circ[(1-\La^{-1})\cdot\phi^{-1}(x) ((B_{\gamma,n}(x))_+\cdot\phi(x))]\circ(1-\La^{-1})^{-1}\circ\phi^{-1})_-\\
    &=&-(\phi\circ(1-\La^{-1})\circ\phi^{-1}(x)\circ (B_{\gamma,n}(x))_+\circ\phi(x)\circ(1-\La^{-1})^{-1}\circ\phi^{-1})_-\\
  &=&-\phi\circ(1-\La^{-1})\circ\phi^{-1}(x)\circ (B_{\gamma,n})_+(x)\circ\phi(x)\circ(1-\La^{-1})^{-1}\circ\phi^{-1}\\
  &&+(\phi\circ(1-\La^{-1})\circ\phi^{-1}(x)\circ B_{\gamma,n}(x)\circ\phi(x)\circ(1-\La^{-1})^{-1}\circ\phi^{-1})_+\\
  &=&-W(B_{\gamma,n})_+W^{-1}+(WB_{\gamma,n}W^{-1})_+.
   \end{eqnarray*}
 Therefore  \[W=\phi\circ(1-\La^{-1})\circ\phi^{-1},\]
 can be as a Darboux transformation of the EBTH.
 Eqs.\eqref{1uN-11}-\eqref{1u-M-11} can be directly got from eq.\eqref{1darbouxL}.

\end{proof}

Define $ \phi_i=\phi_i^{[0]}:=\phi|_{\lambda=\lambda_i}$, then one can choose the specific one-fold  Darboux transformation of the EBTH as following
\[W_1(\lambda_1)=1-\frac{\phi_1}{\La^{-1}\phi_1}\La^{-1}=\frac{\T_1}{\phi_{1}(x-\epsilon)},\]
where

\[\T_1&=&\left|\begin{matrix}1&\Lambda^{-1}\\
\phi_{1}&\phi_{1}(x-\epsilon)
\end{matrix}\right|.\]

Meanwhile, we can also get Darboux transformation on wave function $\phi$ as following

 \[\phi^{[1]}=(1-\frac{\phi_1}{\La^{-1}\phi_1}\La^{-1})\phi.\]
Then using iteration on Darboux transformation, the $j$-th Darboux transformation from the $(j-1)$-th solution is as

\[\phi^{[j]}&=&(1-\frac{\phi_j^{[j-1]}}{\La^{-1}\phi_j^{[j-1]}}\La^{-1})\phi^{[j-1]},\\
u_{N-1}^{[j]}&=&u_{N-1}^{[j-1]}+(\La^N-1)\frac{\phi_j^{[j-1]}}{\La^{-1}\phi_j^{[j-1]}},\\ \notag
\dots&& \dots\\
u_{i}^{[j]}&=&u_{i}^{[j-1]}+u_{i+1}^{[j]}(\La^{i+1}\frac{\phi_j^{[j-1]}}{\La^{-1}\phi_j^{[j-1]}})-\frac{\phi_j^{[j-1]}}{\La^{-1}\phi_j^{[j-1]}}(\La^{-1}u_{i+1}^{[j-1]}),\\ \notag
\dots&& \dots,\\
u_{-M}^{[j]}&=&\frac{\phi_j^{[j-1]}}{\La^{-1}\phi_j^{[j-1]}}(\La^{-1}u_{-M}^{[j-1]})\frac{\La^{-M-1}\phi_j^{[j-1]}}{\La^{-M}\phi_j^{[j-1]}},\]
where $ \phi_i^{[j-1]}:=\phi^{[j-1]}|_{\lambda=\lambda_i},$ are wave functions corresponding to different spectrals with the $(j-1)$-th solutions $u_{N-1}^{[j-1]},u_{N-2}^{[j-1]},\dots,u_{-M}^{[j-1]}.$ It can be checked that $ \phi_i^{[j-1]}=0,\ \ i=1,2,\dots, j-1.$

After iteration on Darboux transformations, the following theorem about the two-fold  Darboux transformation of the EBTH can be derived by direct calculation.

\begin{theorem}
The two-fold  Darboux transformation of the EBTH is as following
\[W_2=1+t_1^{[2]}\Lambda^{-1}+t_2^{[2]}\Lambda^{-2}=\frac{\T_2}{\Delta_2},\]
where

\[\Delta_2=\left|\begin{matrix}
\phi_{1}(x-\epsilon)&\phi_{1}(x-2\epsilon)\\
\phi_{2}(x-\epsilon)&\phi_{2}(x-2\epsilon)
\end{matrix}\right|, \ \ \T_2&=&\left|\begin{matrix}1&\Lambda^{-1}&\Lambda^{-2}\\
\phi_{1}&\phi_{1}(x-\epsilon)&\phi_{1}(x-2\epsilon)\\
\phi_{2}&\phi_{2}(x-\epsilon)&\phi_{2}(x-2\epsilon)
\end{matrix}\right|.\]

The Darboux transformation leads to new solutions from seed solutions
\[u_{N-1}^{[2]}&=&u_{N-1}+(\La^N-1)t_1^{[2]},\\ \notag
\dots&& \dots,\\
u_{i}^{[2]}&=&\sum_{j=i}^{min{(N,i+2)}}u_{j}(x+(i-j)\epsilon)t_{j-i}^{[2]}-\sum_{j=i+1}^{min{(N,i+2)}}u_{j}^{[2]}t_{j-i}^{[2]}(x+j\epsilon),\\ \notag
\dots&& \dots,\\
u_{-M}^{[2]}&=&t_2^{[2]}(x)(\La^{-2}u_{-M})t_2^{[2]-1}(x-M\epsilon),\]
where $u_{i}^{[2]}$ can have another representation as
\[\notag
u_{i}^{[2]}
&=&\frac{\sum_{j=max{(-M,i-2)}}^iu_{j}(x+(i-2-j)\epsilon)t_{j-i+2}^{[2]}-\sum_{j=max{(-M,i-2)}}^{i-1}u_{j}^{[2]}t_{j-i+2}^{[2]}(x+j\epsilon)}
{t_{2}^{[2]}(x+i\epsilon)}.\\\]
\end{theorem}

Similarly, we can generalize the Darboux transformation to $n$-fold case which is contained in the following theorem.

\begin{theorem}\label{ndarboux}
The $n$-fold  Darboux transformation of EBTH equation is as following
\[W_n=1+t_1^{[n]}\Lambda^{-1}+t_2^{[n]}\Lambda^{-2}+\dots+t_{n}^{[n]}\Lambda^{-n}=\frac{1}{\Delta_n}\T_n\]
where

\[ \notag \Delta_n&=&
\left|\begin{matrix}\phi_{1}(x-\epsilon)&\phi_{1}(x-2\epsilon)&\phi_{1}(x-3\epsilon)&\dots&\phi_{1}(x-n\epsilon)\\
\phi_{2}(x-\epsilon)&\phi_{2}(x-2\epsilon)&\phi_{2}(x-3\epsilon)&\dots&\phi_{2}(x-n\epsilon)\\
\phi_{3}(x-\epsilon)&\phi_{3}(x-2\epsilon)&\phi_{3}(x-3\epsilon)&\dots&\phi_{3}(x-n\epsilon)\\
\vdots&\vdots&\vdots&\vdots&\vdots\\
\phi_{n}(x-\epsilon)&\phi_{n}(x-2\epsilon)&\phi_{n}(x-3\epsilon)&\dots&\phi_{n}(x-n\epsilon)\\
\end{matrix}\right|\]
 \[\notag\T_n&=&
\left|\begin{matrix}1&\Lambda^{-1}&\Lambda^{-2}&\Lambda^{-3}&\dots&\Lambda^{-n}\\
\phi_{1}(x)&\phi_{1}(x-\epsilon)&\phi_{1}(x-2\epsilon)&\phi_{1}(x-3\epsilon)&\dots&\phi_{1}(x-n\epsilon)\\
\phi_{2}(x)&\phi_{2}(x-\epsilon)&\phi_{2}(x-2\epsilon)&\phi_{2}(x-3\epsilon)&\dots&\phi_{2}(x-n\epsilon)\\
\phi_{3}(x)&\phi_{3}(x-\epsilon)&\phi_{3}(x-2\epsilon)&\phi_{3}(x-3\epsilon)&\dots&\phi_{3}(x-n\epsilon)\\
\vdots&\vdots&\vdots&\vdots&\vdots\\
\phi_{n}(x)&\phi_{n}(x-\epsilon)&\phi_{n}(x-2\epsilon)&\phi_{n}(x-3\epsilon)&\dots&\phi_{n}(x-n\epsilon)\\\end{matrix}\right|.\]

The Darboux transformation leads to new solutions form seed solutions
\[u_{N-1}^{[n]}&=&u_{N-1}+(\La^N-1)t_1^{[n]},\\ \notag
\dots&& \dots\\
u_{i}^{[n]}&=&\sum_{j=i}^{min{(N,i+n)}}u_{j}(x+(i-j)\epsilon)t_{j-i}^{[n]}-\sum_{j=i+1}^{min{(N,i+n)}}u_{j}^{[n]}t_{j-i}^{[n]}(x+j\epsilon),\\ \notag
\dots&& \dots,\\
u_{-M}^{[n]}&=&t_n^{[n]}(x)(\La^{-n}u_{-M})t_n^{[n]-1}(x-M\epsilon),\]

where $u_{i}^{[n]}$ can have another representation as
\[\notag
u_{i}^{[n]}
&=&\frac{\sum_{j=max{(-M,i-n)}}^iu_{j}(x+(i-n-j)\epsilon)t_{j-i+n}^{[n]}-\sum_{j=max{(-M,i-n)}}^{i-1}u_{j}^{[n]}t_{j-i+n}^{[n]}(x+j\epsilon)}{t_{n}^{[n]}
(x+i\epsilon)}.\\\]
\end{theorem}
It can be easily checked that $W_n\phi_i=0,\ i=1,2,\dots,n.$

Taking seed solution $u_{N-1}=u_{N-2}=\dots =u_{-M+1}=0,u_{-M}=1$, then using Theorem \ref{ndarboux},  one can get the $n$-th new solution of the EBTH as

\[u_{N-1}^{[n]}&=&(\La^{N-1}-\La^{-1})\d_{t_{N-1,0}}\log Wr(\phi_1,\phi_2,\dots\phi_n),\\ \notag
\dots &&\dots\\
u_{-M}^{[n]}&=&e^{(1-\La^{-1})(1-\La^{-M})\log Wr(\phi_1,\phi_2,\dots\phi_n)},\]

where $Wr(\phi_1,\phi_2,\dots\phi_n)$ is the discrete Wronskian
\[Wr(\phi_1,\phi_2,\dots\phi_n)=det (\La^{-j+1} \phi_{n+1-i})_{1\leq i,j\leq n}.\]
Particularly for the $(N,N)$-EBTH, choosing appropriate wave function $\phi$, the $n$-th new solutions can be solitary wave solutions, i.e. $n$-soliton solutions.

\section{Soliton solutions}

After above preparation over the first Darboux transformation, in this section, we will use the first Darboux transformation of the EBTH to generate new solutions from trivial seed solutions. In particular, when $N=M$, some soliton solutions will be shown using the first Darboux transformation.

To give a nice solution, one need to rewrite the extended flows  in the Lax equations of the EBTH in the following lemma.

\begin{lemma}\label{modifiedLax}

The extended flows in Lax formulation of the extended bigraded Toda hierarchy can be equivalently given
by
\begin{equation}
  \label{edef2}
\frac{\partial \L}{\partial t_{-M,n}} = [\bar  A_{-M,n} ,\L ],
\end{equation}
\begin{align}
  &\bar A_{-M,n} = \frac{2}{\epsilon n!} [ \L^n (\log \L - \frac12 ( \frac1M + \frac1N) c_n ) ]_+-\frac{1}{\epsilon n!} [ \L^n (\log_- \L - \frac12 ( \frac1M + \frac1N) c_n ) ],
\end{align}
which can also be rewritten in the form
\begin{align} \notag
  \bar A_{-M,n} &= \frac{1}{\epsilon n!}  \L^n\epsilon \d+\frac{1}{\epsilon n!} [ \L^n ( 2 M \sum_{k < 0} W_k(x) \Lambda^k  - \frac12 ( \frac1M + \frac1N) c_n ) ]_+\\
  &-\frac{1}{\epsilon n!} [ \L^n ( 2 M \sum_{k \geq 0} W_k(x) \Lambda^k  - \frac12 ( \frac1M + \frac1N) c_n ) ]_-.
\end{align}

\end{lemma}
\begin{proof}
Direct calculations will lead to the lemma. Similar results on the ETH can  be seen in \cite{carletthesis}.\end{proof}

Taking seed solution $u_{N-1}=u_{N-2}=\dots =u_{-M+1}=0,u_{-M}=1$, then the initial wave function $\phi_i$ satisfies
\[\label{initial}\Lambda^{N}\phi_i+\Lambda^{-M}\phi_i=\lambda_i\phi_i,\ 1\leq i\leq n.\]

Under this initial equation, the operator $\bar A_{-M,1}$ in above Lemma \ref{modifiedLax} is in form of
\begin{align}\notag
  \bar A_{-M,1}&= \frac{1}{\epsilon }  (\Lambda^{N}+\Lambda^{-M})\epsilon \d+\frac{1}{\epsilon } [ (\Lambda^{N}+\Lambda^{-M}) ( 2 M \sum_{k < 0} W_k(x) \Lambda^k  - \frac12 ( \frac1M + \frac1N)) ]_+\\ \notag
  &-\frac{1}{\epsilon} [ (\Lambda^{N}+\Lambda^{-M}) ( 2 M \sum_{k \geq 0} W_k(x) \Lambda^k  - \frac12 ( \frac1M + \frac1N) ) ]_-\\
  &=\frac{1}{\epsilon }  (\Lambda^{N}+\Lambda^{-M})\epsilon \d-\frac{1}{2\epsilon } ( \frac1M + \frac1N)(\Lambda^{N}-\Lambda^{-M}).
    \end{align}
    Because the coefficients of the Lax operator are all constants,  all the coefficients $W_k(x), k\in \Z$ must be zero.
    Then the linear equation for $t_{-M,1}$  flow is as
\[\label{t-m}\partial_{t_{-M,1}}\phi&=[\frac{1}{\epsilon }  (\Lambda^{N}+\Lambda^{-M})\epsilon \d-\frac{1}{2\epsilon } ( \frac1M + \frac1N)(\Lambda^{N}-\Lambda^{-M})]\phi.
\]
Then solution $\phi$ of eq.\eqref{t-m} in terms of $x,t_{-M,1}$ can be chosen in the form
\begin{align}
\phi&=\exp(\frac{x +\lambda t_{-M,1}}{\epsilon}\log z+\frac{ t_{-M,1}}{2\epsilon}( \frac1M + \frac1N)(-z^N+z^{-M})),
   \end{align}
where $z^N+z^{-M}=\lambda.$ Here the identity   $z^N+z^{-M}=\lambda$ comes from eq.\eqref{initial}.
To include all the dynamical variables, i.e. considering other flows in eq.\eqref{deflax} of the EBTH, wave function $\phi$ can be chosen in the form
\begin{align}\notag
\phi&=z^{\frac{x}{\epsilon}}\exp(\frac{1}{\epsilon}\sum_{k\geq0}t_{-M,k}[\frac{\lambda^k}{k!}\log z+\frac{1 }{2k!}( \frac1M + \frac1N)(-(\lambda^k)_++(\lambda^k)_-)]+ \\
&\sum_{\alpha=1}^{N-1}\sum_{n\geq0}\frac{\Gamma (2- \frac{\alpha}{N} )}{  \epsilon \Gamma(n+2 -\frac{\alpha}{N} ) } ( \lambda^{n+1-\frac{\alpha}N })_+t_{\alpha,n}+\sum_{\beta=-M+1}^0\sum_{n\geq0}\frac{\Gamma (2+\frac{\alpha}{M} )}{  \epsilon \Gamma(n+2 +\frac{\alpha}{M} ) } ( \lambda^{n+1+\frac{\alpha}M })_+t_{\beta,n}),
   \end{align}
where $z^N+z^{-M}=\lambda,$ and the projection $``\pm"$ is about $z$.

Corresponding to the spectral parameters $\lambda$, one can choose  $ \phi$ as
\begin{align}\notag
\phi&=\sum_{m=1}^{N+M}a_{m}\exp (\frac{x +\frac{\lambda^m}{m!} t_{-M,m}}{\epsilon}\log z_{m}+\frac{ t_{-M,m}}{2\epsilon m!}( \frac1M + \frac1N)(-(\lambda_{}^m)_++(\lambda_{}^m)_-)+ \\ \label{phiNM}
&\sum_{\alpha=1}^{N-1}\sum_{n\geq0}\frac{\Gamma (2- \frac{\alpha}{N} )}{  \epsilon \Gamma(n+2 -\frac{\alpha}{N} ) } ( \lambda^{n+1-\frac{\alpha}N })_+t_{\alpha,n}+ \sum_{\beta=-M+1}^0\sum_{n\geq0}\frac{\Gamma (2+\frac{\alpha}{M} )}{  \epsilon \Gamma(n+2 +\frac{\alpha}{M} ) } ( \lambda^{n+1+\frac{\alpha}M })_+t_{\beta,n}),
   \end{align}
   where $z_{m}^N+z_{m}^{-M}=\lambda_,\  1\leq m\leq N+M,$  and the projection $``\pm"$ is about $z_m$. One can find the projection $( \lambda^{n+1+\frac{\alpha}M })_+$ contains infinite number of terms. Therefore this kind of Darboux transformation is not so good to produce new solutions depending on $t_{\beta,n},-M+1\leq \beta\leq 0 $. That is why we will consider another kind of Darboux transformation which benefits in showing the character of the $t_{\beta,n}, -M+1\leq \beta\leq 0$ flow.

When $N=M$, we can choose $ \phi$ as
\begin{align}\notag
\phi&=\sum_{m=1}^{N}a_{m}cosh (\frac{x +\frac{\lambda^m}{m!} t_{-M,m}}{\epsilon}\log z_{m}+\frac{ t_{-M,m}}{2\epsilon m!}( \frac1M + \frac1N)(-(\lambda_{}^m)_++(\lambda_{}^m)_-)+ \\ \label{phiNN}
&\sum_{\alpha=1}^{N-1}\sum_{n\geq0}\frac{\Gamma (2- \frac{\alpha}{N} )}{  \epsilon \Gamma(n+2 -\frac{\alpha}{N} ) } ( \lambda^{n+1-\frac{\alpha}N })_+t_{\alpha,n}+ \sum_{\beta=-M+1}^0\sum_{n\geq0}\frac{\Gamma (2+\frac{\alpha}{M} )}{  \epsilon \Gamma(n+2 +\frac{\alpha}{M} ) } ( \lambda^{n+1+\frac{\alpha}M })_+t_{\beta,n}).
   \end{align}
   Then the generated new solutions of the $(N,M)$-EBTH form trivial solutions using the wave function $\phi$ in eq.\eqref{phiNM} are as
\[\label{uN}u_{N-1}^{[1]}&=&(\La^N-1)\frac{\phi_1}{\La^{-1}\phi_1},\\ \notag
\dots &&\dots\\ \label{uM}
u_{-M}^{[1]}&=&\frac{\phi_1}{\La^{-1}\phi_1}\frac{\La^{-M-1}\phi_1}{\La^{-M}\phi_1}.\]
   The wave function $\phi$ in eq.\eqref{phiNN} will be produced into one-soliton solution of the $(N,N)$-EBTH after transformation eq.\eqref{uN}-eq.\eqref{uM}.
Supposing $N=M=1$, the solutions will be exactly the soliton solutions mentioned in \cite{carletthesis}.
To see the solutions of the EBTH clearly, we will check the equations of  $(2,2)$-EBTH, i.e. $N=M=2$ and prove the solutions after Darboux transformation are also solutions of the $(2,2)$-EBTH.

\subsection{Soliton solutions of the $(2,2)$-EBTH}
 For the $(2,2)$-EBTH, to include all the dynamical variables,
one can choose wave function  $\phi$ in a form of a combination of
\begin{align}\label{phi22}
\phi&=a\bar\phi(z)+b\bar\phi(\frac 1z),
   \end{align}
where
\begin{align}\notag
\bar\phi&=\cosh (\frac{1}{\epsilon}\sum_{k\geq0}[(x+\frac{\lambda^kt_{-2,k}}{k!})\log z+\frac{t_{-2,k} }{2k!}(-(\lambda^k)_++(\lambda^k)_-)]\\ \label{phi22}
&+ \sum_{\alpha=0}^1\sum_{n\geq0}\frac{\Gamma (2- \frac{\alpha}{2} )}{  \epsilon \Gamma(n+2 -\frac{\alpha}{2} ) } ( \lambda^{n+1-\frac{\alpha}2 })_+t_{\alpha,n}+ \sum_{n\geq0}\frac{\Gamma (2-\frac{1}{2} )}{  \epsilon \Gamma(n+2 -\frac{1}{2} ) } ( \lambda^{n+1-\frac{1}2 })_+t_{-1,n}),
   \end{align}
where $z^2+z^{-2}=\lambda,$ and the projection $\pm$ is about the powers of $z$. In the coefficients of $t_{1,n},t_{-1,n}$ in eq.\eqref{phi22}, the square roots of $\lambda$ means the expansion around $\infty$ and $0$ respectively.

Then the soliton solution generated from eq.\eqref{phi22} is as
\[\label{22u1}u_1^{[1]}&=&(\La^2-1)\frac{\phi_1}{\La^{-1}\phi_1},\\
u_{0}^{[1]}&=&u_{1}^{[1]}(\La\frac{\phi_1}{\La^{-1}\phi_1}),\\
u_{-1}^{[1]}&=&u_0^{[1]}(\frac{\phi_1}{\La^{-1}\phi_1}),\\ \label{22u-2}
u_{-2}^{[1]}&=&\frac{\phi_1}{\La^{-1}\phi_1}\frac{\La^{-3}\phi_1}{\La^{-2}\phi_1}.\]

Choose $a=b=1,\epsilon=1$, then

   \begin{align}\label{22simple}
\phi&=\cosh (\xi(x,t,z))+\cosh (\xi(x,t,\frac1z)),\ \ \\
\xi(x,t,z)&:=[(x+(z^2+z^{-2}) t_{-2,1})\log z+\frac{t_{-2,1} }{2}(-z^2+z^{-2})]+  z^2t_{0,0}+zt_{1,0}.
   \end{align}
Bringing $\phi_1=\phi|_{z=z_1}$ in \eqref{22simple} into the above identities \eqref{22u1}-\eqref{22u-2} will lead to specific form of the new solutions

\begin{eqnarray*}u_1^{[1]}&=&(\La^2-1)\frac{\cosh (\xi(x,t,z_1))+\cosh (\xi(x,t,\frac1{z_1}))}{\cosh (\xi(x-1,t,z_1))+\cosh (\xi(x-1,t,\frac1{z_1}))},\\
u_{0}^{[1]}&=&((\La^2-1)\frac{\cosh (\xi(x,t,z_1))+\cosh (\xi(x,t,\frac1{z_1}))}{\cosh (\xi(x-1,t,z_1))+\cosh (\xi(x-1,t,\frac1{z_1}))})\\
&&\times\frac{\cosh (\xi(x+1,t,z_1))+\cosh (\xi(x+1,t,\frac1{z_1}))}{\cosh (\xi(x,t,z_1))+\cosh (\xi(x,t,\frac1{z_1}))},\\
u_{-1}^{[1]}&=&((\La^2-1)\frac{\cosh (\xi(x,t,z_1))+\cosh (\xi(x,t,\frac1{z_1}))}{\cosh (\xi(x-1,t,z_1))+\cosh (\xi(x-1,t,\frac1{z_1}))})\\
&&\times\frac{\cosh (\xi(x+1,t,z_1))+\cosh (\xi(x+1,t,\frac1{z_1}))}{\cosh (\xi(x-1,t,z_1))+\cosh (\xi(x-1,t,\frac1{z_1}))},\\
u_{-2}^{[1]}&=&\frac{\cosh (\xi(x,t,z_1))+\cosh (\xi(x,t,\frac1{z_1}))}{\cosh (\xi(x-1,t,z_1))+\cosh (\xi(x-1,t,\frac1{z_1}))}\frac{\cosh (\xi(x-3,t,z_1))+\cosh (\xi(x-3,t,\frac1{z_1}))}{\cosh (\xi(x-2,t,z_1))+\cosh (\xi(x-2,t,\frac1{z_1}))},
\end{eqnarray*}
whose graphs are as Fig.\ref{u1} in appendix.

One can check that the new solutions satisfy primary flow equations \eqref{10flow},\eqref{00flow},\eqref{-20flow},\eqref{edef2}.
For the above soliton solutions, we can find the velocity of the $t_{-2,1}$ flow is as $-\frac{(z^2+z^{-2}) \log z+\frac{1}{2}(-z^2+z^{-2})}{\log z}$, the velocity of the $t_{0,0}$ flow is as $-\frac{z^2}{\log z}$ and the velocity of the $t_{1,0}$ flow is as $-\frac{z}{\log z}$. From Fig.\ref{V} in appendix, we can see the velocity of $t_{-2,1}$ flow which is quite different from the velocity of the $t_{0,0},t_{1,0}$ flow which can also be seen from the direction of the soliton graphs. When $z$ go to $+\infty$, the velocity becomes $-\infty$, which means the slop of the soliton will approach to be parallel to $x$ axis.

Obviously,  the velocity of the extended flow is quite different from other flows.
Of course, one can get higher order flows in the Lax equation eq.\eqref{edef2}, i.e. $t_{\gamma,n}, n\geq 1$ flows, the new higher-order soliton solutions   and prove that the new solutions and other iterated n-soliton solutions  can also satisfy them which will be not given in detail.

In the next section, it is time to introduce another Darboux transformation of the EBTH basing on another linear equation from eq.\eqref{linearequation}.

\section{The second Darboux transformation of the EBTH}

In fact the EBTH system can also be equivalently rewritten in form of the following linear differential system
\[\label{blinearequation}
\begin{cases}
 \L\psi&=\lambda\psi,\\
\frac{\partial \psi}{\partial t_{\alpha, n}}&=-(B_{\alpha,n})_-\psi,\ \alpha = N-1,N-2, \dots, -M, n \geq 0.\end{cases} \]
We will call the function $\psi$ in eq.\eqref{blinearequation} the second wave function of the EBTH.

In this section, we will consider the Darboux transformation of the EBTH on Lax matrix
 \[\L=\Lambda^{N}+u_{N-1}\Lambda^{N-1}+\dots + u_{-M}
\Lambda^{-M},\]
 i.e.
  \[\label{2dresssing}\L^{[1]}=\Lambda^{N}+u_{N-1}^{[1]}\Lambda^{N-1}+\dots + u_{-M}^{[1]}
\Lambda^{-M}=\bar W\L \bar W^{-1},\]
where $\bar W$ is Darboux transformation operator.

That means after Darboux transformation, the spectral problem

\[\L\psi=\Lambda^{N}\psi+u_{N-1}\Lambda^{N-1}\psi+\dots + u_{-M}
\Lambda^{-M}\psi=\lambda\psi,\]
will become

\[\L^{[1]}\psi^{[1]}=\Lambda^{N}\psi^{[1]}+u_{N-1}^{[1]}\Lambda^{N-1}\psi^{[1]}+\dots + u_{-M}^{[1]}
\Lambda^{-M}\psi^{[1]}=\lambda\psi^{[1]}.\]

To keep the Lax pair eq.\eqref{edef} of the EBTH invariant, i.e.

\begin{equation}
\frac{\partial \L}{\partial t_{\alpha, n}} = [ -(B_{\alpha,n})_- ,\L ],\ \
\frac{\partial \L^{[1]}}{\partial t_{\alpha, n}} = [ -(B_{\alpha,n}^{[1]})_- ,\L^{[1]} ],\ \ B_{\alpha,n}^{[1]}:=B_{\alpha,n}(\L^{[1]}),
\end{equation} the dressing operator $\bar W$ should satisfy following dressing equation
\[\bar W_{t_{\gamma,n}}=\bar W(B_{\gamma,n})_--(\bar WB_{\gamma,n}\bar W^{-1})_-\bar W,\ \ -M \leq \gamma\leq N-1, n\geq 0.\]

\begin{theorem}
If $\psi$ is the second wave function of the EBTH,
the second Darboux transformation operator of the EBTH
 \[\bar W(\lambda)=(\frac{\La\psi}{\psi}-\La)=\psi(x+\epsilon)\circ(1-\La)\circ\psi^{-1}(x),\]

will generater new solutions $u_{i}^{[1]}$ from seed solutions
$u_{i}, -M\leq i\leq N-1$

\[u_{N-1}^{[1]}&=&\La u_{N-1}+(\La^N-1)\frac{\La\psi}{\psi},\\ \notag
\dots&& \dots\\
u_{i}^{[1]}&=&\La u_{i}+u_{i+1}^{[1]}(\La^{i+1}\frac{\La\psi}{\psi})-\frac{\La\psi}{\psi}(\La u_{i+1}),\\ \notag
\dots&& \dots,\\
u_{-M}^{[1]}&=&\frac{\La\psi}{\psi} u_{-M}\frac{\La^{-M}\psi}{\La^{-M+1}\psi}.\]

\end{theorem}
\begin{proof}

Using eq.\eqref{Cpos} in Lemma \ref{lema}, a direct computation will lead to the following
 \begin{eqnarray*}\bar W_{t_{\gamma,n}}\bar W^{-1}&=&(\psi(x+\epsilon)\circ(1-\La)\circ\psi^{-1})_{t_{\gamma,n}}\psi\circ(1-\La)^{-1}\circ\psi^{-1}(x+\epsilon)\\
 &=&-(((B_{\gamma,n})_-\psi(x+\epsilon))\circ(1-\La)\circ\psi^{-1})\psi\circ(1-\La)^{-1}\circ\psi^{-1}(x+\epsilon)\\
 &&+\psi(x+\epsilon)\circ(1-\La)\circ((B_{\gamma,n})_-\psi)\psi^{-1}\circ(1-\La)^{-1}\circ\psi^{-1}(x+\epsilon)\\
 &=&-((B_{\gamma,n})_-\psi(x+\epsilon))\psi^{-1}(x+\epsilon)\\
 &&+\psi(x+\epsilon)\circ(1-\La)\circ((B_{\gamma,n})_-\psi)\psi^{-1}\circ(1-\La)^{-1}\circ\psi^{-1}(x+\epsilon)\\
&=&-\psi(x+\epsilon)\circ [(\La-1)\cdot(((B_{\gamma,n})_-\psi)\psi^{-1})]\La\circ(1-\La)^{-1}\circ\psi^{-1}(x+\epsilon)\\
  &=&-(\psi(x+\epsilon)\circ[(\La-1)\cdot(\psi^{-1}(B_{\gamma,n})_-\circ\psi)]\La\circ(1-\La)^{-1}\circ\psi^{-1}(x+\epsilon))_+\\
    &=&\psi(x+\epsilon)\circ(1-\La)\circ\psi^{-1}(x)\circ (B_{\gamma,n})_-(x)\circ\psi(x)\circ(1-\La)^{-1}\circ\psi^{-1}(x+\epsilon)\\
  &&-(\psi(x+\epsilon)\circ(1-\La)\circ\psi^{-1}(x)\circ B_{\gamma,n}(x)\circ\psi(x)\circ(1-\La)^{-1}\circ\psi^{-1}(x+\epsilon))_-\\
  &=&\bar W(B_{\gamma,n})_-\bar W^{-1}-(\bar WB_{\gamma,n}\bar W^{-1})_-.
   \end{eqnarray*}
 Therefore  \[\bar W=\psi(x+\epsilon)\circ(1-\La)\circ\psi^{-1},\]
 can be as another Darboux transformation of the EBTH.
 The new solutions can be got easily using the second dressing form eq.\eqref{2dresssing}.

\end{proof}

Define $ \psi_i=\psi_i^{[0]}:=\psi|_{\lambda=\lambda_i}$, then one can choose the specific one-fold  Darboux transformation of the EBTH equations as following
\[\bar W_1(\lambda_1)=\frac{\psi_{1}(x+\epsilon)}{\psi_{1}(x)}-\La=\frac{\T_1}{\psi_{1}(x)},\]
where

\[\T_1&=&\left|\begin{matrix}1&\Lambda\\
\psi_{1}&\psi_{1}(x+\epsilon)
\end{matrix}\right|.\]

Meanwhile, we can also get the Darboux transformation on wave function $\psi$ as following

 \[\psi^{[1]}=(\frac{\La\psi_1}{\psi_1}-\La)\psi.\]
Then using iteration on Darboux transformation, the $j$-th Darboux transformation from the $(j-1)$-th solution is as

\[\psi^{[j]}&=&(\frac{\La\psi_j^{[j-1]}}{\psi_j^{[j-1]}}-\La)\psi^{[j-1]},\\
u_{N-1}^{[j]}&=&\La u_{N-1}^{[j-1]}+(\La^N-1)\frac{\La\psi_j^{[j-1]}}{\psi_j^{[j-1]}},\\ \notag
\dots&& \dots\\
u_{i}^{[j]}&=&\La u_{i-1}^{[j-1]}+u_{i+1}^{[j]}(\La^{i+1}\frac{\La\psi_j^{[j-1]}}{\psi_j^{[j-1]}})-\frac{\La\psi_j^{[j-1]}}{\psi_j^{[j-1]}}u_{i+1}^{[j-1]},\\ \notag
\dots&& \dots,\\
u_{-M}^{[j]}&=&\frac{\La\psi_j^{[j-1]}}{\psi_j^{[j-1]}} u_{-M}^{[j-1]}\frac{\La^{-M}\psi_j^{[j-1]}}{\La^{-M+1}\psi_j^{[j-1]}},\]
where $ \psi_i^{[j-1]}:=\psi^{[j-1]}|_{\lambda=\lambda_i},$ are wave functions corresponding to different spectral with the $(j-1)$-th solutions $u_{N-1}^{[j-1]},u_{N-2}^{[j-1]},\dots,u_{-M}^{[j-1]}.$ It can be checked that $ \psi_i^{[j-1]}=0,\ \ i=1,2,\dots, j-1.$

\begin{theorem}
The two-fold  Darboux transformation of the EBTH is as following
\[\bar W_2=t_0^{[2]}+t_1^{[2]}\Lambda+\Lambda^{2}=\frac{\T_2}{\Delta_2},\]
where

\[\Delta_2=\left|\begin{matrix}
\psi_{1}(x)&\psi_{1}(x+\epsilon)\\
\psi_{2}(x)&\psi_{2}(x+\epsilon)
\end{matrix}\right|, \ \ \T_2&=&\left|\begin{matrix}1&\Lambda&\Lambda^{2}\\
\psi_{1}&\psi_{1}(x+\epsilon)&\psi_{1}(x+2\epsilon)\\
\psi_{2}&\psi_{2}(x+\epsilon)&\psi_{2}(x+2\epsilon)
\end{matrix}\right|.\]

The Darboux transformation leads to new solutions from seed solutions
\[u_{N-1}^{[2]}&=&\La^2u_{N-1}-(\La^N-1)t_1^{[2]},\\ \notag
\dots&& \dots,\\
u_{i}^{[2]}&=&\sum_{j=i}^{min{(N,i+2)}}u_{j}(x+(i+2-j)\epsilon)t_{i+2-j}^{[2]}-\sum_{j=i+1}^{min{(N,i+2)}}u_{j}^{[2]}t_{i+2-j}^{[2]}(x+j\epsilon),\\ \notag
\dots&& \dots,\\
u_{-M}^{[2]}&=&t_0^{[2]}(x)u_{-M}t_0^{[2]-1}(x-M\epsilon),\]
where $u_{i}^{[2]}$ can have another representation as
\[
u_{i}^{[2]}
&=&\frac{\sum_{j=max{(-M,i-2)}}^iu_{j}(x+(i-j)\epsilon)t_{i-j}^{[2]}-\sum_{j=max{(-M,i-2)}}^{i-1}u_{j}^{[2]}t_{i-j}^{[2]}(x+j\epsilon)}{t_{0}^{[2]}(x+i\epsilon)}.\]
\end{theorem}

Similarly, we can generalize the Darboux transformation to $n$-fold case which is contained in the following theorem.

\begin{theorem}\label{bndarboux}
The $n$-fold  Darboux transformation of EBTH equation is as following
\[\bar W_n=t_0^{[n]}+t_1^{[n]}\Lambda+t_2^{[n]}\Lambda^{2}+\dots+(-1)^n\Lambda^{n}=\frac{1}{\Delta_n}\T_n\]
where

\[ \notag \Delta_n&=&
\left|\begin{matrix}\psi_{1}(x)&\psi_{1}(x+\epsilon)&\psi_{1}(x+2\epsilon)&\dots&\psi_{1}(x+(n-1)\epsilon)\\
\psi_{2}(x)&\psi_{2}(x+\epsilon)&\psi_{2}(x+3\epsilon)&\dots&\psi_{2}(x+(n-1)\epsilon)\\
\psi_{3}(x)&\psi_{3}(x+\epsilon)&\psi_{3}(x+3\epsilon)&\dots&\psi_{3}(x+(n-1)\epsilon)\\
\vdots&\vdots&\vdots&\vdots&\vdots\\
\psi_{n}(x)&\psi_{n}(x+\epsilon)&\psi_{n}(x+3\epsilon)&\dots&\psi_{n}(x+(n-1)\epsilon)\\
\end{matrix}\right|\]
 \[\notag\T_n&=&
\left|\begin{matrix}1&\Lambda&\Lambda^{2}&\Lambda^{3}&\dots&\Lambda^{n}\\
\psi_{1}(x)&\psi_{1}(x+\epsilon)&\psi_{1}(x+2\epsilon)&\psi_{1}(x+3\epsilon)&\dots&\psi_{1}(x+n\epsilon)\\
\psi_{2}(x)&\psi_{2}(x+\epsilon)&\psi_{2}(x+2\epsilon)&\psi_{2}(x+3\epsilon)&\dots&\psi_{2}(x+n\epsilon)\\
\psi_{3}(x)&\psi_{3}(x+\epsilon)&\psi_{3}(x+2\epsilon)&\psi_{3}(x+3\epsilon)&\dots&\psi_{3}(x+n\epsilon)\\
\vdots&\vdots&\vdots&\vdots&\vdots\\
\psi_{n}(x)&\psi_{n}(x+\epsilon)&\psi_{n}(x+2\epsilon)&\psi_{n}(x+3\epsilon)&\dots&\psi_{n}(x+n\epsilon)\\\end{matrix}\right|.\]

The Darboux transformation leads to new solutions form seed solutions
\[u_{N-1}^{[n]}&=&\La^nu_{N-1}+(-1)^n(1-\La^N)t_{n-1}^{[n]},\\ \notag
\dots&& \dots\\
u_{i}^{[n]}&=&\sum_{j=i}^{min{(N,i+n)}}u_{j}(x+(i+n-j)\epsilon)t_{i+n-j}^{[n]}-\sum_{j=i+1}^{min{(N,i+n)}}u_{j}^{[n]}t_{i+n-j}^{[n]}(x+j\epsilon),\\ \notag
\dots&& \dots,\\
u_{-M}^{[n]}&=&t_0^{[n]}(x)u_{-M}t_0^{[n]-1}(x-M\epsilon),\]

where $u_{i}^{[n]}$ can have another representation as
\[
u_{i}^{[n]}
&=&\frac{\sum_{j=max{(-M,i-n)}}^iu_{j}(x+(i-j)\epsilon)t_{i-j}^{[n]}-\sum_{j=max{(-M,i-n)}}^{i-1}u_{j}^{[n]}t_{i-j}^{[n]}(x+j\epsilon)}{t_{0}^{[n]}(x+i\epsilon)}.\]
\end{theorem}
It can be easily checked that $\bar W_n\psi_i=0,\ i=1,2,\dots,n.$

Taking seed solution $u_{N-1}=u_{N-2}=\dots =u_{-M+1}=0,u_{-M}=1$, then using Theorem \ref{bndarboux},  one can get the $n$-th new solution of the EBTH as

\[u_{N-1}^{[n]}&=&(-1)^n(1-\La^N)\d_{t_{N-1,0}}\log Wr(\psi_1,\psi_2,\dots\psi_n),\\ \notag
\dots &&\dots\\
u_{-M}^{[n]}&=&e^{(\La-1)(1-\La^{-M})\log Wr(\psi_1,\psi_2,\dots\psi_n)},\]

where $Wr(\psi_1,\psi_2,\dots\psi_n)$ is the discrete Wronskian
\[Wr(\psi_1,\psi_2,\dots\psi_n)=det (\La^{j-1} \psi_{n+1-i})_{1\leq i,j\leq n}.\]
Particularly for the $(N,N)$-EBTH, choosing appropriate wave function $\psi$, another family of $n$-th new solutions can be solitary wave solutions, i.e. $n$-soliton solutions.

\section{Another kind of soliton solutions}

In this section, we will use the second Darboux transformation of the EBTH to generate a new kind of  solutions from trivial seed solutions. In particular, when $N=M$, some soliton solutions will be shown using Darboux transformation.

To give a nice solution, one need to rewrite the extended flow  in the Lax equations of the EBTH in the following form
\begin{equation}
  \label{edef2b}
\frac{\partial \L}{\partial t_{-M,n}} = [\bar  A_{-M,n} ,\L ],
\end{equation}
\begin{align}
  &\bar  A_{-M,n}  = -\frac{2}{\epsilon n!} [ \L^n (\log \L - \frac12 ( \frac1M + \frac1N) c_n ) ]_-+\frac{1}{\epsilon n!} [ \L^n (\log_+ \L - \frac12 ( \frac1M + \frac1N) c_n ) ].
\end{align}

Taking seed solution $u_{N-1}=u_{N-2}=\dots =u_{-M+1}=0,u_{-M}=1$, then the initial wave function $\psi_i$ satisfies
\[\label{binitial}\Lambda^{N}\psi_i+\Lambda^{-M}\psi_i=\lambda_i\psi_i,\ 1\leq i\leq n.\]

To include all the dynamical variables, i.e. considering other flows in eq.\eqref{deflax} of the EBTH, wave function $\psi$ can be chosen in the form as
\begin{align}\notag
\psi&=\sum_{m=1}^{N+M}a_{m}\exp (\frac{x +\frac{\lambda^m}{m!} t_{-M,m}}{\epsilon}\log z_{m}+\frac{ t_{-M,m}}{2\epsilon m!}( \frac1M + \frac1N)(-(\lambda_{}^m)_++(\lambda_{}^m)_-)- \\ \label{bphiNM}
&\sum_{\alpha=1}^{N-1}\sum_{n\geq0}\frac{\Gamma (2- \frac{\alpha}{N} )}{  \epsilon \Gamma(n+2 -\frac{\alpha}{N} ) } ( \lambda^{n+1-\frac{\alpha}N })_-t_{\alpha,n}- \sum_{\beta=-M+1}^0\sum_{n\geq0}\frac{\Gamma (2+\frac{\alpha}{M} )}{  \epsilon \Gamma(n+2 +\frac{\alpha}{M} ) } ( \lambda^{n+1+\frac{\alpha}M })_-t_{\beta,n}),
   \end{align}
   where $z_{m}^N+z_{m}^{-M}=\lambda,\  1\leq m\leq N+M,$ and the projection $``\pm"$ is about $z_m$. One can find the projection $( \lambda^{n+1+\frac{\alpha}M })_+$ now contains finite number of terms. Therefore this kind of Darboux transformation benefits in showing the character of the $t_{\beta,n}, -M+1\leq \beta\leq 0$ flow.  But one can find the projection $ ( \lambda^{n+1-\frac{\alpha}N })_-$ now contains infinite number of terms which tell us the second Darboux transformation is not so good to produce new solutions depending on $t_{\alpha,n},1\leq \alpha\leq N-1$. From this point, the extended bigraded Toda hierarchy is quite different from the extended Toda hierarchy\cite{CDZ}.
 Here the identity   $z_{m}^N+z_{m}^{-M}=\lambda$ comes from eq.\eqref{binitial}.

When $N=M$, we can choose $ \psi$ as
\begin{align}\notag
\psi&=\sum_{m=1}^{N}a_{m}cosh (\frac{x +\frac{\lambda^m}{m!} t_{-M,m}}{\epsilon}\log z_{m}+\frac{ t_{-M,m}}{2\epsilon m!}( \frac1M + \frac1N)(-(\lambda_{}^m)_++(\lambda_{}^m)_-)- \\ \label{bphiNN}
&\sum_{\alpha=1}^{N-1}\sum_{n\geq0}\frac{\Gamma (2- \frac{\alpha}{N} )}{  \epsilon \Gamma(n+2 -\frac{\alpha}{N} ) } ( \lambda^{n+1-\frac{\alpha}N })_-t_{\alpha,n}- \sum_{\beta=-M+1}^0\sum_{n\geq0}\frac{\Gamma (2+\frac{\alpha}{M} )}{  \epsilon \Gamma(n+2 +\frac{\alpha}{M} ) } ( \lambda^{n+1+\frac{\alpha}M })_-t_{\beta,n}).
   \end{align}
   Then the new generated new solutions of the $(N,M)$-EBTH form trivial solutions using the wave function $\psi$ in eq.\eqref{bphiNM} are as
\[\label{buN}u_{N-1}^{[1]}&=&(\La^N-1)\frac{\psi_1}{\La\psi_1},\\ \notag
\dots &&\dots\\ \label{buM}
u_{-M}^{[1]}&=&\frac{\La\psi_1}{\psi_1}\frac{\La^{-M}\psi_1}{\La^{-M+1}\psi_1}.\]
   The wave function $\psi$ in eq.\eqref{bphiNN} will be produced into one-soliton solution of the $(N,N)$-EBTH after transformation eq.\eqref{buN}-eq.\eqref{buM}.
Supposing $N=M=1$, the solutions will be exactly another kind of soliton solutions which are different from the soliton solutions mentioned in \cite{carletthesis}.
To see the solutions of the EBTH clearly, we will check the equations of  $(2,2)$-EBTH, i.e. $N=M=2$ and prove the new kind of solutions after Darboux transformation are also solutions of the $(2,2)$-EBTH.

 For the $(2,2)$-EBTH, to include all the dynamical variables,
one can choose wave function  $\psi$ in a form of a combination of
\begin{align}\label{bphi22}
\psi&=a\bar\psi(z)+b\bar\psi(\frac 1z),
   \end{align}
where
\begin{align}\notag
\bar\psi&=\cosh (\frac{1}{\epsilon}\sum_{k\geq0}[(x+\frac{\lambda^kt_{-2,k}}{k!})\log z+\frac{t_{-2,k} }{2k!}(-(\lambda^k)_++(\lambda^k)_-)]\\
&- \sum_{\alpha=0}^1\sum_{n\geq0}\frac{\Gamma (2- \frac{\alpha}{2} )}{  \epsilon \Gamma(n+2 -\frac{\alpha}{2} ) } ( \lambda^{n+1-\frac{\alpha}2 })_-t_{\alpha,n}- \sum_{n\geq0}\frac{\Gamma (2-\frac{1}{2} )}{  \epsilon \Gamma(n+2 -\frac{1}{2} ) } ( \lambda^{n+1-\frac{1}2 })_-t_{-1,n}),
   \end{align}
where $z^2+z^{-2}=\lambda,$ and the projection $\pm$ is about the powers of $z$. In the coefficients of $t_{1,n},t_{-1,n}$, the square roots of $\lambda$ mean the expansions around $\infty$ and $0$ respectively.

Then the soliton solution generated from eq.\eqref{bphi22} is as
\[\label{b22u1}u_1^{[1]}&=&(\La^2-1)\frac{\La\psi_1}{\psi_1},\\
u_{0}^{[1]}&=&u_{1}^{[1]}(\La\frac{\La\psi_1}{\psi_1}),\\
u_{-1}^{[1]}&=&u_0^{[1]}(\frac{\La\psi_1}{\psi_1}),\\ \label{b22u-2}
u_{-2}^{[1]}&=&\frac{\La\psi_1}{\psi_1}\frac{\La^{-2}\psi_1}{\La^{-1}\psi_1}.\]

Choose $a=b=1,\epsilon=1$, then

   \begin{align}\label{b22simple}
\psi&=\cosh (\xi(x,t,z))+\cosh (\xi(x,t,\frac1z)),\ \ \\
\xi(x,t,z)&:=[(x+(z^2+z^{-2}) t_{-2,1})\log z+\frac{t_{-2,1} }{2}(-z^2+z^{-2})]- z^{-2}t_{0,0}-z^{-1}t_{-1,0}.
   \end{align}
Bringing $\psi_1=\psi|_{z=z_1}$ in \eqref{b22simple} into the above identities \eqref{b22u1}-\eqref{b22u-2} will lead to specific form of the new solutions

\begin{eqnarray*}u_1^{[1]}&=&(\La^2-1)\frac{\cosh (\xi(x+1,t,z_1))+\cosh (\xi(x+1,t,\frac1{z_1}))}{\cosh (\xi(x,t,z_1))+\cosh (\xi(x,t,\frac1{z_1}))},\\
u_{0}^{[1]}&=&((\La^2-1)\frac{\cosh (\xi(x+1,t,z_1))+\cosh (\xi(x+1,t,\frac1{z_1}))}{\cosh (\xi(x,t,z_1))+\cosh (\xi(x,t,\frac1{z_1}))})\\
&&\frac{\cosh (\xi(x+2,t,z_1))+\cosh (\xi(x+2,t,\frac1{z_1}))}{\cosh (\xi(x+1,t,z_1))+\cosh (\xi(x+1,t,\frac1{z_1}))},\\
u_{-1}^{[1]}&=&((\La^2-1)\frac{\cosh (\xi(x+1,t,z_1))+\cosh (\xi(x+1,t,\frac1{z_1}))}{\cosh (\xi(x,t,z_1))+\cosh (\xi(x,t,\frac1{z_1}))})\\
&&\frac{\cosh (\xi(x+2,t,z_1))+\cosh (\xi(x+2,t,\frac1{z_1}))}{\cosh (\xi(x,t,z_1))+\cosh (\xi(x,t,\frac1{z_1}))},\\
u_{-2}^{[1]}&=&\frac{\cosh (\xi(x+1,t,z_1))+\cosh (\xi(x+1,t,\frac1{z_1}))}{\cosh (\xi(x,t,z_1))+\cosh (\xi(x,t,\frac1{z_1}))}\frac{\cosh (\xi(x-2,t,z_1))+\cosh (\xi(x-2,t,\frac1{z_1}))}{\cosh (\xi(x-1,t,z_1))+\cosh (\xi(x-1,t,\frac1{z_1}))}.
\end{eqnarray*}

One can check that the new solutions satisfy primary flow equations \eqref{00flow},\eqref{-10flow},\eqref{-20flow},\eqref{edef2b}.
For the above soliton solutions, we can find the velocity of the $t_{-2,1}$ flow is the same as the soliton solution generated from the first Darboux transformation, the velocity of the $t_{0,0}$ flow is as $\frac{1}{z^{2}\log z}$ and the velocity of the $t_{-1,0}$ flow is as $\frac{1}{z\log z}$. From Fig.\ref{Vanother} in appendix, we can see the velocity of $t_{-2,1}$ flow which is quite different from the velocity of the $t_{0,0},t_{-1,0}$ flow which can also be seen from the direction of the soliton graphs.  When $z$ go to $+\infty$, the velocity of the $t_{0,0},t_{-1,0}$ becomes $0$, which means the slop of the $t_{0,0},t_{-1,0}$  flow of the soliton flows  will also approach to be vertical to $x$ axis.

\sectionnew{Conclusions and Discussions} In this paper, we give two kinds of Darboux transformation of the EBTH which are used to generate new solutions from seed solutions including soliton solutions of $(N,N)$-EBTH. Meanwhile we plotted the soliton graphs of the $(N,N)$-EBTH from which some approximation analysis is done. From the analysis on velocities of soliton solutions, the effect of the extended flows are shown. Meanwhile the higher order Darboux transformation and the general form of the $n$-th new solutions are also given.   It was found that the velocities of the extended flows are quite different from other flows. The application of these above two Darboux transformations in generating new solutions from known solutions related to the Gromov-Witten theory of orbiford $c_{NM}$ is an interesting questions. Also how to generate higher order solutions using the adjoint Darboux transformation and Binary Darboux transformation (see \cite{binamatveev,Oevel}) of the EBTH and what  the relation between the above two Darboux transformation in this paper and the adjoint and Binary Darboux transformation is are both interesting questions.

\vskip 0.5truecm \noindent{\bf Acknowledgments.}
 This work is supported by  the Zhejiang Provincial Natural Science Foundation under Grant No. LY15A010004, the National Natural Science Foundation of China under Grant No. 11571192,
 the Natural Science Foundation of Ningbo under Grant No. 2015A610157 and the K. C. Wong Magna Fund in
Ningbo University. The author would like to thank Professor J. S. He for his helpful suggestions.

{\bf Appendix:}

\begin{figure}[h!]
\centering
\raisebox{0.85in}{($u_1^{[1]}$)}\includegraphics[scale=0.30]{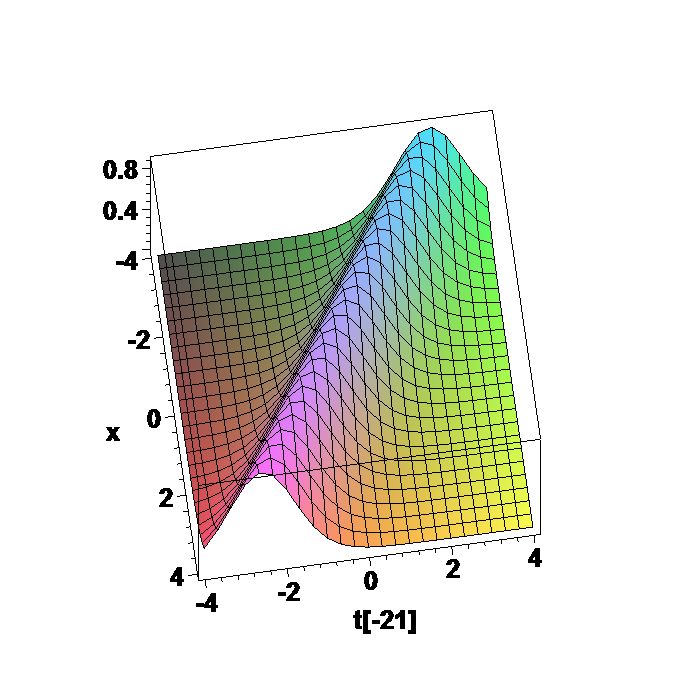}
\hskip 0.03cm
\raisebox{0.85in}{($u_0^{[1]}$)}\raisebox{-0.1cm}{\includegraphics[scale=0.30]{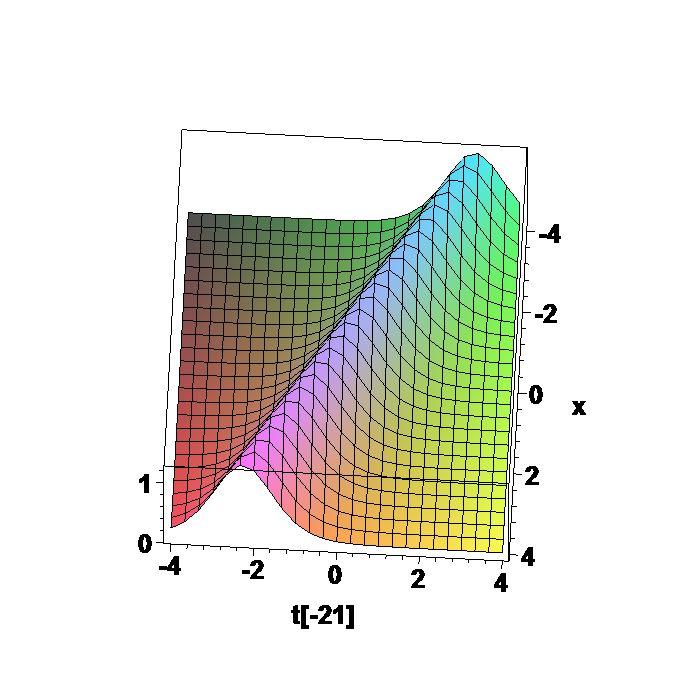}}
\hskip 0.03cm
\raisebox{0.85in}{($u_{-1}^{[1]}$)}\raisebox{-0.1cm}{\includegraphics[scale=0.30]{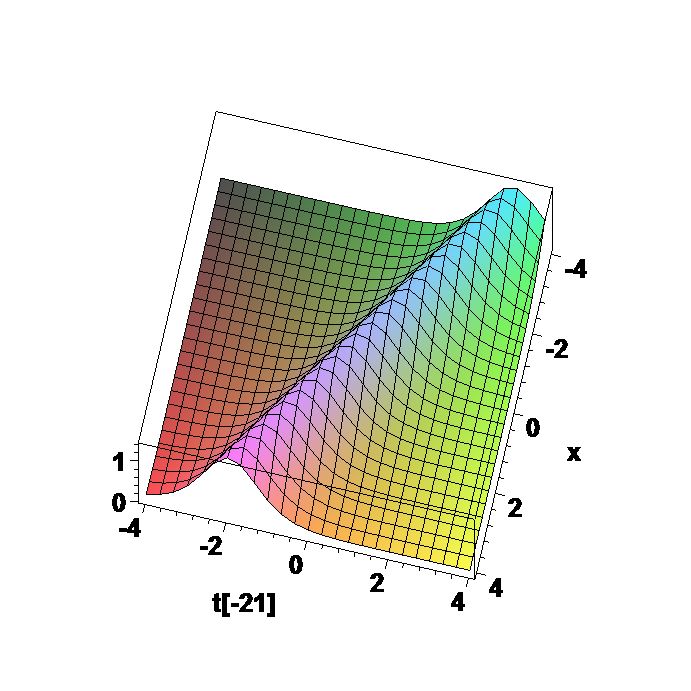}}
\hskip 0.03cm
\raisebox{0.85in}{($u_{-2}^{[1]}$)}\raisebox{-0.1cm}{\includegraphics[scale=0.30]{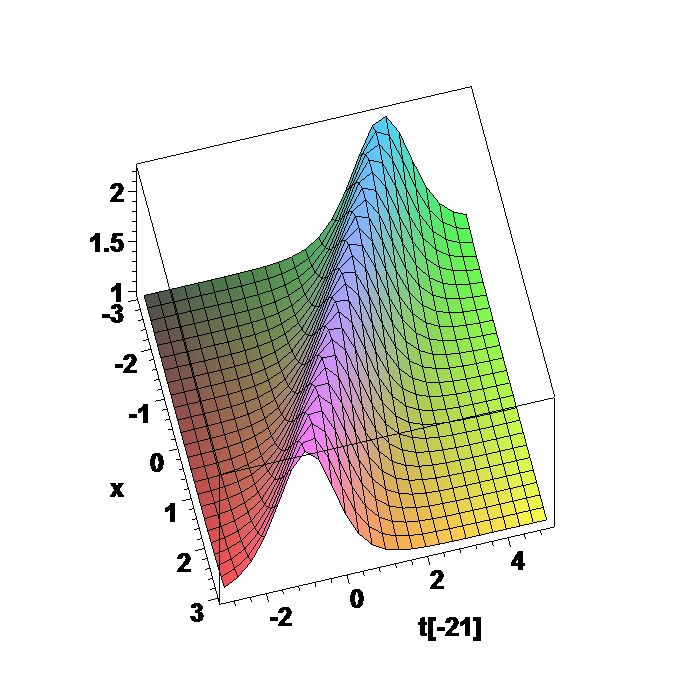}}
 \caption{\small (color online)Soliton solution $u_1^{[1]},u_{0}^{[1]},u_{-1}^{[1]},u_{-2}^{[1]}$  of the $(2,2)$-EBTH when $z=2,a=b=1$ where the label
   $t[-21]$ is denoted as $t_{-2,1}$.} \label{u1}
\end{figure}

\begin{figure}[h!]
\centering
\raisebox{0.85in}{($v_{t_{-2,1}}$)}\includegraphics[scale=0.17]{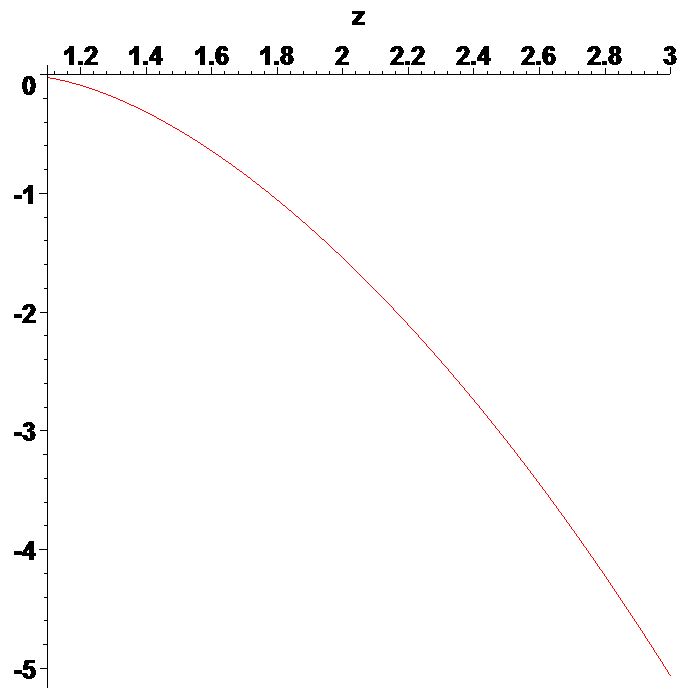}
\hskip 0.03cm
\raisebox{0.85in}{($v_{t_{0,0}}$)}\raisebox{-0.1cm}{\includegraphics[scale=0.15]{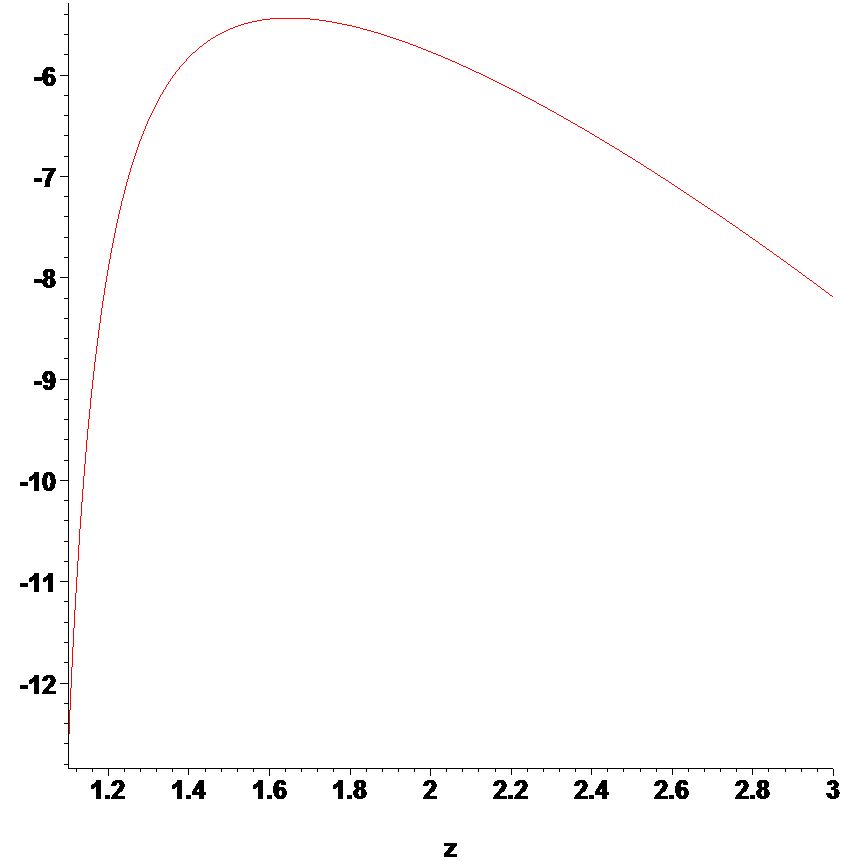}}
\hskip 0.03cm
\raisebox{0.85in}{($v_{t_{1,0}}$)}\raisebox{-0.1cm}{\includegraphics[scale=0.15]{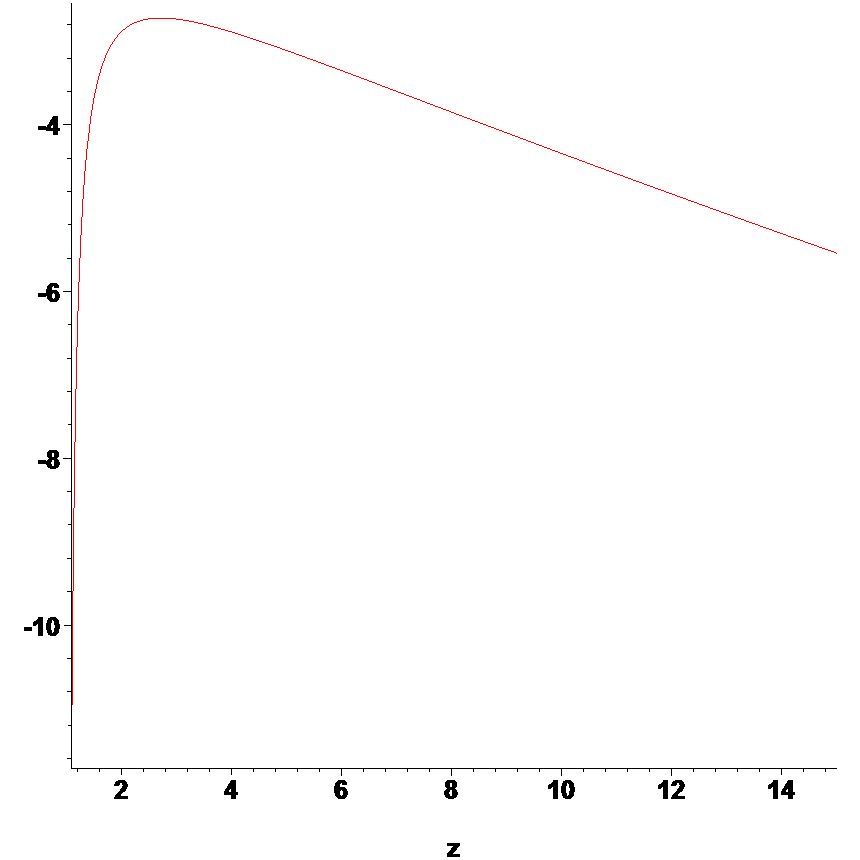}}
 \caption{\small (color online)The velocity plot of the $t_{-2,1}$ flow, $t_{0,0}$ flow, $t_{1,0}$ flow of the soliton solution generated from the first Darboux transformation  the $(2,2)$-EBTH.} \label{V}
\end{figure}

\begin{figure}[h!]
\centering
\raisebox{0.85in}{($v_{t_{0,0}}$)}\raisebox{-0.1cm}{\includegraphics[scale=0.20]{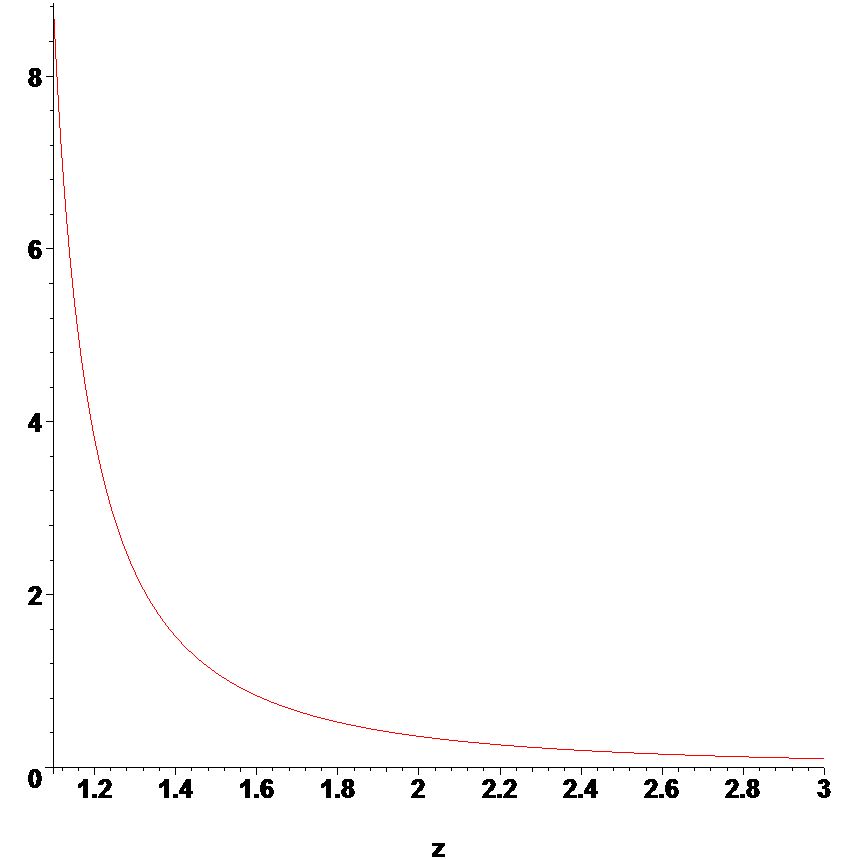}}
\hskip 0.03cm
\raisebox{0.85in}{($v_{t_{-1,0}}$)}\raisebox{-0.1cm}{\includegraphics[scale=0.20]{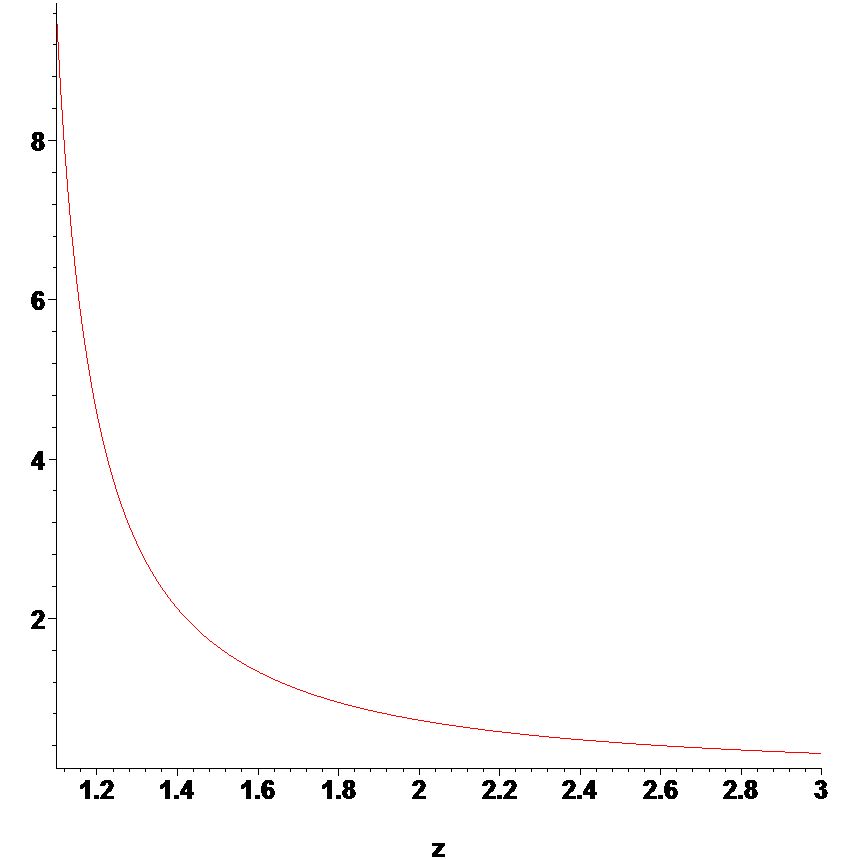}}
 \caption{\small (color online)The velocity plot of the  $t_{0,0}$ flow, $t_{-1,0}$ flow of another soliton solution generated from the second Darboux transformation of the $(2,2)$-EBTH.} \label{Vanother}
\end{figure}



\end{document}